\documentclass[letterpaper, 10 pt, conference]{ieeeconf}  

\IEEEoverridecommandlockouts                              
\usepackage{graphicx} 
\usepackage{float}    
\usepackage{verbatim} 
\usepackage{amsmath}  
\usepackage{amssymb}  
\usepackage{amsthm}
\usepackage[font=footnotesize]{caption}
\usepackage[ruled,vlined,linesnumbered]{algorithm2e}
\usepackage{multirow}
\usepackage{array}
\usepackage{booktabs}
\usepackage{hyperref}
\usepackage{gensymb}
\usepackage[dvipsnames]{xcolor}

\newtheorem{theorem}{Theorem}

\theoremstyle{definition}
\newtheorem{assumption}{Assumption}

\theoremstyle{definition}
\newtheorem{definition}{Definition}

\title{\LARGE \bf
A Sequential Quadratic Programming Approach to the Solution of Open-Loop Generalized Nash Equilibria
}

\author{Edward L. Zhu and Francesco Borrelli
\thanks{This work was supported by the National Science Foundation under Grant No. 1931853.}
\thanks{Edward L. Zhu, and Francesco Borrelli {\{\tt\small edward.zhu, fborrelli\}@berkeley.edu} are with the Department of Mechanical Engineering at the University of California, Berkeley CA, USA}
}

\begin{document}

\maketitle
\thispagestyle{empty}
\pagestyle{empty}

\begin{abstract}

Dynamic games can be an effective approach to modeling interactive behavior between multiple non-cooperative agents and they provide a theoretical framework for simultaneous prediction and control in such scenarios. In this work, we propose a numerical method for the solution of local generalized Nash equilibria (GNE) for the class of open-loop general-sum dynamic games for agents with nonlinear dynamics and constraints. In particular, we formulate a sequential quadratic programming (SQP) approach which requires only the solution of a single convex quadratic program at each iteration. Central to the robustness of our approach is a non-monotonic line search method and a novel merit function for SQP step acceptance. We show that our method achieves linear convergence in the neighborhood of local GNE and we derive an update rule for the merit function which helps to improve convergence from a larger set of initial conditions. We demonstrate the effectiveness of the algorithm in the context of car racing, where we show up to 32\% improvement of success rate when comparing against a state-of-the-art solution approach for dynamic games. \url{https://github.com/zhu-edward/DGSQP}.

\end{abstract}

\section{Introduction}

Robotic systems are being deployed in the real world with increasing regularity, as such they must be able to operate safely in an environment which is inhabited by other intelligent agents who are each pursuing their own agenda. In the process of doing so, it is likely that interactions between the agents arise due to limitations on the shared workspace. How the robot handles such interactions is critical to its performance and safety as these scenarios typically involve inter-agent constraints becoming active. The primary challenge here is that of information availability, where agents, either due to technological limitations or to maintain a competitive advantage, do not necessarily share information with each other about their intentions or future plans. It is therefore crucial to model the behavior of the other agents in the environment when planning the actions of the robot.

Traditional approaches to this problem adopt a pipeline architecture where an upstream module provides forecasts of the behavior of other agents \cite{betz2022autonomous,montemerlo2008junior}. These forecasts are then treated as constant during the downstream planning phase when an action plan is formulated for the robot. The drawback of such an architecture is that the effects of the robot's actions are not considered when constructing the predictions of the behavior of the other agents, which can be problematic as interactions typically affect all of the agents involved. It is therefore desirable to \emph{simultaneously} tackle the task of prediction and planning in order to better capture the effect of interactions on all agents in the environment.

Dynamic non-cooperative games \cite{bacsar1998dynamic} provide a theoretical framework for simultaneous prediction and planning for multiple ego-centric agents and has seen many applications in trajectory optimization for robotic systems \cite{cleac2020algames,fridovich2020efficient,schwarting2019social,laine2021multi}. This is done by formulating a set of coupled optimization problems which describe the behavior of an agent as a function of the others' in the environment. Two solution or equilibrium concepts are common for dynamic games, namely the Stackelberg and Nash equilibria, which make different assumptions on the structure of the game. A Stackelberg equilibrium can be found for a game with an explicit leader-follower hierarchy \cite{liniger2019noncooperative}, whereas a Nash equilibrium models the case when agents make their decisions simultaneously. Our work focuses on the selection of Nash equilibria, which we believe to be a good fit for modeling the behavior of ego-centric agents when no \emph{a priori} structure is imposed on the order of the interactions. More specifically, we are interested in generalized Nash equilibria (GNE) \cite{facchinei2010generalized} for dynamic games with state and input constraints.

In this work, we propose an iterative approach for finding GNE of a discrete-time dynamic game based on sequential quadratic programming (SQP), which builds upon ideas from \cite{laine2021computation}. In particular, we are able to handle general nonlinear game dynamics and constraints on both the game state and agent actions. By leveraging a non-monotone line search strategy with a novel merit function for step selection, our approach is able to converge to a solution in more cases of our simulation study when compared to a state-of-the-art augmented Lagrangian approach \cite{cleac2020algames}. Our primary contributions are the following:
\begin{enumerate}
    \item A solver for GNE of dynamic games with nonlinear dynamics and constraints.
    \item Proof of local convergence for the SQP based approach.
    \item A novel merit function, which when used in conjunction with a non-monotone line search strategy greatly improves solver convergence in practice.
    \item A simulation study in the context of car racing which shows up to 32\% improvement in success rate when comparing our approach with a state-of-the-art method.
\end{enumerate}

\noindent \emph{Related Work:} Many approaches to game-theoretic trajectory optimization have been proposed in the literature. \cite{fridovich2020efficient, schwarting2021stochastic, kavuncu2021potential} take a differential dynamic programming approach to obtain a linear quadratic approximation of the dynamic game. However, the approach is unable to explicitly account for inequality constraints and instead must include them in the cost function via barrier functions. This can obfuscate the meaning of the game as the cost measures both performance and constraint violation. \cite{spica2020real,wang2021game} formulate a decomposition based method called Iterative Best Response where the agents improve their strategy in a sequential manner while holding the behavior of all other agents fixed. It is shown that fixed points of this algorithm correspond to GNE. However, the method requires the solution of the same number of optimization problems as there are agents and can be slow to converge in practice. In contrast, our approach only requires the solution of a single optimization problem at each iteration. The approach proposed in this work is inspired by and builds upon \cite{cleac2020algames} and \cite{laine2021computation} with some major improvements which are demonstrated to be essential to effectively solve the class of problems we are interested in. \cite{laine2021computation} proposes a similar SQP approach, but does not investigate its convergence properties. Compared to \cite{laine2021computation}, we also introduce a new merit function and line search strategy which improves solver convergence. \cite{cleac2020algames} proposes a GNE solver based on an augmented Lagrangian approach. The solver, called ALGAMES, shows good  performance when compared to \cite{fridovich2020efficient}. However, as will be shown in our comparison, ALGAMES appears to struggle with convergence in the context of car racing where more complex dynamics and environments are introduced.

\section{Problem Formulation}

Consider an $M$-agent, finite-horizon, discrete-time, general-sum, open-loop, dynamic game whose state is characterized by the joint dynamical system:
\begin{align} \label{eq:joint_dynamics}
    x_{k+1} = f(x_k, u_k),
\end{align}
where $x_k^i \in \mathcal{X}^i$ and $u_k^i \in \mathcal{U}^i$ are the state and input of agent $i$ at time step $k$  and
\begin{align*}
    x_k &= \begin{bmatrix} {x_k^1}^\top, \dots, {x_k^M}^\top \end{bmatrix}^\top \in \mathcal{X}^1 \times \dots \times \mathcal{X}^M = \mathcal{X} \subseteq \mathbb{R}^n \\
    u_k &= \begin{bmatrix} {u_k^1}^\top, \dots, {u_k^M}^\top \end{bmatrix}^\top \in \mathcal{U}^1 \times \dots \times \mathcal{U}^M = \mathcal{U} \subseteq  \mathbb{R}^m,
\end{align*}
are the concatenated states and inputs of all agents. In this work, we will use the notation $x_k^{\neg i}$ and $u_k^{\neg i}$ to denote the vector of states and inputs of all but the $i$-th agent.

Each agent $i$ attempts to minimize its own cost function, which is comprised of stage costs $l_k^i$ and terminal cost $l_N^i$, over a horizon of length $N$:
\begin{subequations} \label{eq:agent_cost}
\begin{align} 
    \bar{J}^i(\mathbf{x}, \mathbf{u}^i) &= \sum_{k = 0}^{N-1} l_k^i(x_k, u_k^i) + l_N^i(x_N) \label{eq:agent_cost_xu} \\
    &= J^i(\mathbf{u}^1, \dots, \mathbf{u}^M, x_0), \label{eq:agent_cost_u}
\end{align}
\end{subequations}
where $\mathbf{x} = \{x_0, \dots, x_N\}$ and $\mathbf{u}^i = \{u_0^i, \dots, u_{N-1}^i\}$ denote state and input sequences over the horizon. Note that the cost in \eqref{eq:agent_cost_xu} for agent $i$ depends on its \emph{own} inputs and the \emph{joint} state. We arrive at \eqref{eq:agent_cost_u} by recursively substituting in the dynamics \eqref{eq:joint_dynamics} to the cost function, which are naturally a function of the open-loop input sequences for all agents. The agents are additionally subject to $n_c$ constraints
\begin{align} \label{eq:joint_constraints}
    C(\mathbf{u}^1, \dots, \mathbf{u}^M, x_0) \leq 0,
\end{align}
which can be used to describe individual constraints as well as coupling between agents and where we have once again made the dependence on the joint dynamics implicit. For the sake of brevity, when focusing on agent $i$, we omit the inital state $x_0$ and write the cost and constraint functions as $J^i(\mathbf{u}^i,\mathbf{u}^{\neg i})$ and $C(\mathbf{u}^i,\mathbf{u}^{\neg i})$. Let us now define the conditional constraint set
\begin{align*}
    \mathcal{U}^i(\mathbf{u}^{\neg i}) =& \{ \mathbf{u}^i \ | \ C(\mathbf{u}^i, \mathbf{u}^{\neg i}) \leq 0 \},
\end{align*}
which can be interpreted as a restriction of the joint constraint set for agent $i$ given some $\mathbf{u}^{\neg i}$. We make the following assumption about the functions and constraint sets.
\begin{assumption} \label{asm:compact_set_differentiability}
    The sets $\mathcal{X}^i$ and $\mathcal{U}^i$ are compact and the functions $f$, $J^i$, and $C$ are twice continuously differentiable on $\mathcal{X}$ and $\mathcal{U}$ for all $i \in \{1, \dots, M\}$.
\end{assumption}

\subsection{Generalized Nash Equilibrium}

We define the constrained dynamic game as the tuple:
\begin{align} \label{eq:dynamic_game}
    \Gamma = (N, \mathcal{X}, \mathcal{U}, f, \{J^i\}_{i=1}^M, C).
\end{align}
For such a game, a GNE is attained at the set of feasible input sequences $\mathbf{u} = \{\mathbf{u}^{i}\}_{i=1}^M$ which minimize \eqref{eq:agent_cost} for all agents $i$. Formally, we define this solution concept as follows:
\begin{definition}
    A generalized Nash equilibrium (GNE) \cite{facchinei2010generalized} for the dynamic game $\Gamma$ is the set of open-loop solutions $\mathbf{u}^{\star} = \{\mathbf{u}^{i,\star}\}_{i=1}^M$ such that for each agent $i$:
    \begin{align}
        J^i(\mathbf{u}^{i,\star}, \mathbf{u}^{\neg i,\star}) \leq J^i(\mathbf{u}^{i}, \mathbf{u}^{\neg i,\star}), \ \forall \mathbf{u}^{i} \in \mathcal{U}^i(\mathbf{u}^{\neg i,\star}). \nonumber
    \end{align}
    If the condition holds only in some local neighborhood of $\mathbf{u}^{i,\star}$, then $\mathbf{u}^{\star}$ is denoted as  a local GNE.
\end{definition}
In other words, at a local GNE, agents cannot improve their cost by unilaterally perturbing their open-loop solution in a locally feasible direction.  The local GNE for agent $i$ can be obtained equivalently by solving the following constrained finite horizon optimal control problems (FHOCP):
\begin{align} \label{eq:agent_fhocp}
    \mathbf{u}^{i,\star}(\mathbf{u}^{\neg i,\star}) = \arg\min_{\mathbf{u}^i} \ & \ J^i(\mathbf{u}^i, \mathbf{u}^{\neg i,\star}) \\
    \text{subject to} \ & \ C(\mathbf{u}^i, \mathbf{u}^{\neg i,\star}) \leq 0. \nonumber
\end{align}
where $\mathbf{u}^{\neg i,\star}$ correspond to local GNE solutions for the other agents. Note that we are assuming uniqueness of the local GNE of \eqref{eq:dynamic_game}. This assumption will be made formal in the next section. A distinct advantage of using \eqref{eq:agent_fhocp} to model agent interactions is that a dynamic game allows for a direct representation of agents with competing objectives as the $M$ objectives are considered separately instead of being summed together, which is typical in cooperative multi-agent approaches \cite{zhu2020trajectory}.

\section{An SQP Approach to Dynamic Games} \label{sec:sqp_approach}

In this section, we propose a method which iteratively solves for open-loop local GNE of dynamic games using sequential quadratic approximations. In particular, we will derive the algorithm and present guarantees on local convergence, which is based on established SQP theory \cite{boggs1995sequential}. We begin by defining the Lagrangian functions for the $M$ coupled FHOCPs in \eqref{eq:agent_fhocp}:
\begin{align*}
    \mathcal{L}^i(\mathbf{u}^i, \mathbf{u}^{\neg i,\star}, \lambda^i)  = J^i(\mathbf{u}^i, \mathbf{u}^{\neg i,\star}) +  C(\mathbf{u}^i, \mathbf{u}^{\neg i,\star})^\top \lambda^i,
\end{align*}
where we have again omitted the dependence on the initial state $x_0$ for brevity. As in \cite{cleac2020algames}, we require that the Lagrange multipliers $\lambda^i \geq 0$ are equal over all agents, i.e $\lambda^i = \lambda^j = \lambda$, $\forall i,j \in \{1, \dots,M\}$. Since the multipliers reflect the sensitivity of the optimal cost w.r.t. constraint violation, this can be interpreted as a requirement for parity in terms of the cost of constraint violation for all agents. Under this condition, the GNE from \eqref{eq:agent_fhocp} are also known as normalized Nash equilibria \cite{rosen1965existence}.

A direct consequence of writing the constrained dynamic game in the coupled nonlinear optimization form of \eqref{eq:agent_fhocp} is that, subject to regularity conditions, solutions of \eqref{eq:agent_fhocp} must satisfy the KKT conditions below:
\begin{subequations} \label{eq:kkt}
\begin{align} 
    \nabla_{\mathbf{u}^i} \mathcal{L}^i(\mathbf{u}^{i,\star}, \mathbf{u}^{\neg i,\star}, \lambda^\star) &= 0, \ \forall i = 1, \dots, M, \label{eq:stationarity}\\
    C(\mathbf{u}^{1,\star},\dots, \mathbf{u}^{M,\star}) &\leq 0, \label{eq:primal_feasibility}\\
    C(\mathbf{u}^{1,\star},\dots, \mathbf{u}^{M,\star})^{\top}\lambda^\star &= 0, \\
    \lambda^\star & \geq 0.
\end{align}
\end{subequations}

We therefore propose to find a local GNE as a solution to the KKT system \eqref{eq:kkt} in an iterative fashion starting from an initial guess for the primal and dual solution, which we denote as $\mathbf{u}^i_0$ and $\lambda_0 \geq 0$ respectively, and taking steps $p_q^i$ and $p_q^\lambda$, at iteration $q$, to obtain the sequence of iterates:
\begin{align} \label{eq:sqp_step}
    \mathbf{u}^i_{q+1} = \mathbf{u}^i_q + p_q^i, \ \lambda_{q+1} = \lambda_q + p_q^\lambda.
\end{align}
In particular, we form a quadratic approximation of \eqref{eq:stationarity} and linearize the constraints in \eqref{eq:primal_feasibility} about the primal and dual solution at iteration $q$ in a SQP manner \cite{wright1999numerical} as follows:
\begin{align} \label{eq:sqp_approximation}
    L_q &= \begin{bmatrix}
    \nabla_{\mathbf{u}^1}^2 \mathcal{L}_q^1 & \nabla_{\mathbf{u}^2,\mathbf{u}^1} \mathcal{L}_q^1 & \dots & \nabla_{\mathbf{u}^M,\mathbf{u}^1} \mathcal{L}_q^1 \\
    \nabla_{\mathbf{u}^1,\mathbf{u}^2} \mathcal{L}_q^2 & \nabla_{\mathbf{u}^2}^2 \mathcal{L}_q^2 & \dots & \nabla_{\mathbf{u}^M,\mathbf{u}^2} \mathcal{L}_q^2 \\
    \vdots & \vdots & \ddots & \vdots \\
    \nabla_{\mathbf{u}^1,\mathbf{u}^M} \mathcal{L}_q^M & \nabla_{\mathbf{u}^2,\mathbf{u}^M} \mathcal{L}_q^M & \dots & \nabla_{\mathbf{u}^M}^2 \mathcal{L}_q^M
    \end{bmatrix}, \nonumber\\
    h_q &= \begin{bmatrix} \nabla_{\mathbf{u}^1} J_q^1 & \nabla_{\mathbf{u}^2} J_q^2 & \dots & \nabla_{\mathbf{u}^M} J_q^M \end{bmatrix}^\top, \nonumber\\
    G_q &= \begin{bmatrix} \nabla_{\mathbf{u}^1} C_q & \nabla_{\mathbf{u}^2} C_q & \dots & \nabla_{\mathbf{u}^M} C_q \end{bmatrix}, \nonumber\\
    B_q &= \text{proj}_{\succeq 0}((L_q+L_q^\top)/2) + \epsilon I,
\end{align}
where the subscript $q$ indicates that the corresponding quantity is evaluated at the primal and dual iterate $\mathbf{u}_q$ and $\lambda_q$. Here, $\epsilon \geq 0$ is a regularization coefficient, $I$ is the identity matrix of appropriate size, and $\text{proj}_{\succeq 0} = \sum_{i=1}^n \max\{0, s_i\}v_iv_i^\top$ denotes the operation which projects the symmetric matrix $X\in\mathbb{R}^{n\times n}$ onto the positive semi-definite cone, where $s_i$ and $v_i$ denote the $i$-th eigenvalue and eigenvector of $X$ respectively. Using this approximation, we solve for the step in the primal variables via the following convex quadratic program (QP):
\begin{subequations} \label{eq:sqp_qp}
\begin{align} 
    p_q^\mathbf{u} = \arg \min_{p^1, \dots, p^M} \ & \ \frac{1}{2} {p^\mathbf{u}}^\top B_q {p^\mathbf{u}} + h_q^\top {p^\mathbf{u}} \label{eq:sqp_qp_cost}\\
    \text{subject to} \ & \ C_q + G_q {p^\mathbf{u}} \leq 0. \label{eq:sqp_qp_ineq_constraint}
\end{align}
\end{subequations}
where we denote $p^\mathbf{u} = [{p^1}^\top, \dots, {p^M}^\top]^\top$. Denote the Lagrange multipliers corresponding to the solution of \eqref{eq:sqp_qp} as $d_q$. We then define the step in the dual variables as 
\begin{align} \label{eq:dual_step}
    p_q^\lambda = d_q - \lambda_q,
\end{align}
which maintains nonnegativity of the dual iterates given $\lambda_0 \geq 0$. We note that in contrast to the approach used in \cite{spica2020real} and \cite{wang2021game}, which require the solution of $M$ optimization problems, our SQP procedure requires the solution of only a single QP at each iteration.

\subsection{Local Behavior of Dynamic Game SQP}

We make the following assumptions about the primal and dual solutions of \eqref{eq:agent_fhocp}:
\begin{assumption} \label{asm:optimality}
    Solutions $\{\mathbf{u}^{i,\star}\}_{i=1}^M$ and $\lambda^\star$ of \eqref{eq:agent_fhocp} satisfy the following, for each $i\in\{1,\dots,M\}$:
    \begin{itemize}
        \item $\lambda^\star \perp C(\mathbf{u}^{i,\star}, \mathbf{u}^{\neg i,\star})$ and $\lambda_j^\star = 0 \iff C_j(\mathbf{u}^{i,\star}, \mathbf{u}^{\neg i,\star}) < 0$ for $j=1,\dots,n_c$.
        \item The rows of the Jacobian of the active constraints at the local GNE, i.e. $\nabla_{\mathbf{u}^i}\bar{C}(\mathbf{u}^{i,\star}, \mathbf{u}^{\neg i,\star})$, are linearly independent.
        \item $d^\top \nabla_{\mathbf{u}^i}^2 \mathcal{L}^i(\mathbf{u}^{i,\star}, \mathbf{u}^{\neg i,\star}, \lambda^\star) d > 0$, $\forall d \neq 0$ such that $\nabla_{\mathbf{u}^i}\bar{C}(\mathbf{u}^{i,\star}, \mathbf{u}^{\neg i,\star})^\top d = 0$.
    \end{itemize}
\end{assumption}
The first assumption is strict complementary slackness, the second is the linear independence constraint qualification (LICQ), and the third states that the Hessian of the Lagrangian function is positive definite on the null space of the active constraint Jacobians at the solution. It is straight forward to see that \eqref{eq:kkt} and Assumption~\ref{asm:optimality} together constitute necessary and sufficient conditions for a primal and dual solution of \eqref{eq:agent_fhocp} for agent $i$ to be locally optimal and unique. When these conditions hold for the solutions over all agents, satisfaction of the requirements for a unique local GNE follow immediately. This result was proven formally in \cite{laine2021computation}. Note that, as in \cite{boggs1995sequential} and \cite{laine2021computation}, Assumption~\ref{asm:optimality} is standard and can be verified a posteriori.

To analyze the local behavior of the iterative procedure as defined by \eqref{eq:sqp_step}, \eqref{eq:sqp_approximation}, and \eqref{eq:sqp_qp}, let us assume that $\mathbf{u}_0$ and $\lambda_0$ are close to the optimal solution and the subset of active constraints at the local GNE, which we denote as $\bar{C}$, with Jacobian $\bar{G}$, is known and constant at each iteration $q$. Therefore, for the purposes of this section, we can replace the inequality constraint in \eqref{eq:sqp_qp_ineq_constraint} with the equality constraint:
\begin{align} \label{eq:sqp_qp_eq_constraint}
    \bar{C}_q + \bar{G}_q p^{\mathbf{u}} = 0.
\end{align}
We refer to the QP constructed from \eqref{eq:sqp_qp_cost} and \eqref{eq:sqp_qp_eq_constraint} as EQP.

In the traditional derivation of the SQP procedure \cite{boggs1995sequential,wright1999numerical}, it was shown that under the aforementioned assumptions, the SQP step computed is identical to a Newton step for the corresponding KKT system. The SQP step therefore inherits the quadratic convergence rate of Newton's method in a local neighborhood of the optimal solution \cite[Theorem 3.1]{boggs1995sequential}. However, in the case of dynamic games, the equivalence between the SQP procedure and Newton's method is no longer exact since the matrix $L_q$ is not symmetric in general. To see this, let us first state the joint KKT system for the equality constrained version of \eqref{eq:agent_fhocp}:
\begin{align} \label{eq:root_finding_problem}
    & F(\mathbf{u}^{1,\star}, \dots, \mathbf{u}^{M,\star}, \lambda^\star) \\
    & \qquad = \begin{bmatrix}
        \nabla_{\mathbf{u}^1} \mathcal{L}^1(\mathbf{u}^{1,\star},\dots, \mathbf{u}^{M,\star}, \lambda^\star) \\
        \vdots \\
        \nabla_{\mathbf{u}^M} \mathcal{L}^M(\mathbf{u}^{1,\star},\dots, \mathbf{u}^{M,\star}, \lambda^\star) \\
        \bar{C}(\mathbf{u}^{1,\star},\dots, \mathbf{u}^{M,\star})
    \end{bmatrix} = 0, \nonumber
\end{align}
For the system of equations \eqref{eq:root_finding_problem}, the Newton step at iteration $q$ is the solution of the linear system:
\begin{align} \label{eq:newton_step}
    \begin{bmatrix}
    L_q & \bar{G}_q^\top \\
    \bar{G}_q & 0
    \end{bmatrix} \begin{bmatrix} \bar{p}_q^\mathbf{u} \\ \bar{p}_q^\lambda \end{bmatrix} = - \begin{bmatrix} h_q + \bar{G}_q^\top \lambda_q \\ \bar{C}_q \end{bmatrix}.
\end{align}
On the other hand, by the first order optimality conditions for EQP, we have that the SQP step must satisfy
\begin{align} \label{eq:sqp_kkt}
    \begin{bmatrix}
    B_q & \bar{G}_q^\top \\
    \bar{G}_q & 0
    \end{bmatrix} \begin{bmatrix} p_q^\mathbf{u} \\ p_q^\lambda \end{bmatrix} = - \begin{bmatrix} h_q + \bar{G}_q^\top \lambda_q \\ \bar{C}_q \end{bmatrix}.
\end{align}
When the matrix $L_q$ is positive definite and $\epsilon=0$, \eqref{eq:sqp_kkt} and \eqref{eq:newton_step} are equivalent. This corresponds to the special case of potential games \cite{zhu2008lagrangian}. However, this is not true in general for our SQP step, which implies that we cannot inherit the quadratic convergence of Newton's method. Instead, the SQP step from \eqref{eq:sqp_qp_cost} and \eqref{eq:sqp_qp_eq_constraint} can be seen as a symmetric approximation to the Newton step. As such, we establish guaranteed local linear convergence for our SQP procedure via established theory for SQP with approximate Hessians. Before proving the main result of this section, let us first define the bounded deterioration property, which essentially requires that the distance between a matrix and its approximations are bounded.
\begin{definition} \label{def:bounded_deterioration}
    A sequence of matrix approximations $\{B_q\}$ to $L^\star$ for the SQP method is said to have the property of \emph{bounded deterioration} if there exist constants $\alpha_1$ and $\alpha_2$ independent of $q$ such that:
    \begin{align*}
        \|B_{q+1}-L^\star\| \leq (1+\alpha_1 \sigma_q) \|B_q-L^\star\| + \alpha_2 \sigma_q,
    \end{align*}
    where $\sigma_q = \max(\|\mathbf{u}_{q+1}-\mathbf{u}^\star\|, \|\mathbf{u}_{q}-\mathbf{u}^\star\|, \|\lambda_{q+1}-\lambda^\star\|, \|\lambda_{q}-\lambda^\star\|)$, and $L^\star$ denotes the matrix $L$ in \eqref{eq:sqp_approximation} evaluated at the solution $\mathbf{u}^\star$, $\lambda^\star$.
\end{definition}
\begin{theorem}
Consider the dynamic game defined by \eqref{eq:dynamic_game}. Let Assumptions~\ref{asm:compact_set_differentiability} and \ref{asm:optimality} hold. Then there exist positive constants $\epsilon_1$ and $\epsilon_2$ such that if
\begin{align*}
    \|\mathbf{u}_0 - \mathbf{u}^\star\| \leq \epsilon_1, \ \|B_0 - L^\star\| \leq \epsilon_2,
\end{align*}
and $\lambda_0 = -(\bar{G}_0 \bar{G}_0^\top)^{-1}\bar{G}_0 h_0$, then the sequence $(\mathbf{u}_q, \lambda_q)$ generated by the SQP procedure \eqref{eq:sqp_step} and \eqref{eq:sqp_qp} converges linearly to $(\mathbf{u}^\star, \lambda^\star)$.
\end{theorem}
\begin{proof}
We obtain the result by showing that the conditions of \cite[Theorem 3.3]{boggs1995sequential} are satisfied. The first condition requires that the approximations $B_q$ are positive definite on the null space of $\bar{G}_q$. Since $B_q$ is constructed by projecting $L_q$ into the positive definite cone, this condition is satisfied trivially. 

The second condition requires that the sequence $\{B_q\}$ satisfies the bounded deterioration property from Definition~\ref{def:bounded_deterioration}. To show this, we begin with the following derivation:
\begin{align*}
    \|B_{q+1}-L^\star\| &= \|B_{q+1}-B_q + B_q - L^\star\| \\
    & \leq \|B_q - L^\star\| + \|B_{q+1}-B_q\|.
\end{align*}
From the above, it can be seen that $\{B_q\}$ satisfies the bounded deterioration property with $\alpha_1 = 0$ if there exists $\alpha_2$ such that $\|B_{q+1}-B_q\| \leq \alpha_2 \sigma_q$. For the equality constrained QP defined by \eqref{eq:sqp_qp_cost} and \eqref{eq:sqp_qp_eq_constraint}, we can write the dual solution at iteration $q$ analytically as $d_q = (\bar{G}_q B_q^{-1} \bar{G}_q^\top)^{-1}(\bar{C}_q - \bar{G}_q B_q^{-1} h_q)$ and the primal solution as $p_q^\mathbf{u}=-B_q^{-1}(h_q + \bar{G}_q^\top d_q)$, which by Assumption~\ref{asm:compact_set_differentiability} are both continuous in $\mathbf{u}_q$ and $\lambda_q$. Therefore, setting $\mathbf{u}_{q+1} = \mathbf{u}_{q} + p_q^\mathbf{u}$, $\lambda_{q+1} = d_q$, and noting that $B_q$ is also continuous in $\mathbf{u}_q$ and $\lambda_q$, we have that there exists $\beta_1$ and $\beta_2$ such that
\begin{align*}
    \|B_{q+1}-B_q\| &\leq \beta_1 \|\mathbf{u}_{q+1}-\mathbf{u}_{q}\| + \beta_2 \|\lambda_{q+1} - \lambda_{q}\| \\
    &\leq \beta_1 (\|\mathbf{u}_{q+1}-\mathbf{u}^\star\|+\|\mathbf{u}_{q}-\mathbf{u}^\star\|) \\
    &\quad + \beta_2 (\|\lambda_{q+1} - \lambda^\star\| + \|\lambda_{q} - \lambda^\star\|) \\
    &\leq \text{max}(\beta_1, \beta_2) \sigma_q,
\end{align*}
which satisfies the bounded deterioration property with $\alpha_2 = \text{max}(\beta_1, \beta_2)$. Given Assumption~\ref{asm:optimality} holds, linear convergence of the SQP iterations follows directly as a consequence of \cite[Theorem 3.3]{boggs1995sequential}.
\end{proof}

\section{A Novel Merit Function and Non-monotone Line Search Strategy for Dynamic Game SQP} \label{sec:practical_considerations}

We have shown that our proposed SQP approach exhibits linear convergence when close to a local GNE. However, as is commonly seen with numerical methods for nonlinear optimization, a na{\"i}ve implementation of the procedure defined by \eqref{eq:sqp_step}, \eqref{eq:sqp_approximation}, and \eqref{eq:sqp_qp} often performs poorly due to overly aggressive steps leading to diverging iterates. In this section, we introduce a merit function and line search method which will help address this problem in practice. These components will be used to determine how much of the SQP step $p_q^\mathbf{u}$ and $p_q^\lambda$ can be taken to make progress towards a local GNE while remaining in a region about the current iterate where the QP approximation \eqref{eq:sqp_qp} is valid. In the numerical examples that we propose at the end of the paper, this modification is crucial in improving the practicality of the proposed SQP approach when given initial guesses beyond the immediate neighborhood of a local GNE.

\subsection{A Novel Merit Function}

Merit functions are commonly used in conjunction with a backtracking line search technique as a mechanism for measuring progress towards the optimal solution and limiting the step size taken by an iterative solver \cite{wright1999numerical}. This is important as the quality of the QP approximation \eqref{eq:sqp_qp} may be poor when far away from the iterate, which could result in divergence of the iterates if large steps are taken.

In traditional constrained optimization, merit functions typically track a combination of cost value and constraint violation. In the context of dynamic games, this is not as straightforward as the agents may have conflicting objectives and a proposed step may result in an increase in the objectives of some agents along with a decrease in others'. We also cannot just simply sum the objectives as minimizers of the combined cost function may not be local GNE. We therefore propose the following merit function:
\begin{align} \label{eq:merit_function}
    \phi(\mathbf{u},&\lambda, s;\mu) = \nonumber \\
    &\frac{1}{2} \| \underbrace{\begin{bmatrix} \nabla_{\mathbf{u}^1} \mathcal{L}^1(\mathbf{u}, \lambda) \\
        \vdots \\
        \nabla_{\mathbf{u}^M} \mathcal{L}^M(\mathbf{u}, \lambda) \end{bmatrix}}_{=\nabla \mathcal{L}(\mathbf{u},\lambda)} \|_2^2  + \mu \|C(\mathbf{u}) - s\|_1, 
\end{align}
and define $\gamma(\mathbf{u},\lambda) = (1/2)\|\nabla \mathcal{L}(\mathbf{u},\lambda)\|_2^2$. The slack variable $s = \text{min}(0, C(\mathbf{u}))$ is defined element-wise such that $C-s$ captures violation of the inequality constraints and we define the step $p^s = C(\mathbf{u}) + G(\mathbf{u})p^\mathbf{u} - s$. Compared to the merit function from \cite{laine2021computation}, which only included the first term of \eqref{eq:merit_function}, the novelty of our's is the $l^1$ norm term, whose purpose will be described shortly. Instead of measuring the agent objectives, our merit function tracks the first order optimality conditions in addition to constraint violation. It should be easy to see that the merit function attains a minimum of zero at any local GNE. However, we note that this merit function is not \emph{exact} \cite{wright1999numerical} since the first order conditions are only necessary for optimality.

Since we would like the sequence of iterates to converge to the minimizers of $\phi$, it follows that at each iteration we would like take a step in a descent direction of $\phi$. As such, we will now, in a manner similar to \cite{wright1999numerical}, analyze the directional derivative of the proposed merit function in the direction of the steps computed from \eqref{eq:sqp_qp} and describe a procedure for choosing the parameter $\mu$ and the corresponding \emph{conditions} such that a descent in $\phi$ is guaranteed.

Since the $l^1$ norm is not differentiable everywhere, we begin by taking a Taylor series expansion of $\gamma$ and $C-s$ for $\alpha \in (0,1]$:
\begin{align*}
    &\phi(\mathbf{u}+\alpha p^{\mathbf{u}},\lambda+\alpha p^\lambda, s+\alpha p^s;\mu) - \phi(\mathbf{u},\lambda, s;\mu) \\
    &\leq \alpha \nabla_{\mathbf{u},\lambda}\gamma\begin{bmatrix} p^{\mathbf{u}} \\ p^\lambda \end{bmatrix} + \mu\|C-s + \alpha(Gp^\mathbf{u}-p^s)\|_1 \nonumber \\
    & \qquad - \mu\|C-s\|_1 + \beta \alpha^2 \|p\|_2^2 \\
    & = \alpha\left(\nabla_{\mathbf{u},\lambda}\gamma\begin{bmatrix} p^{\mathbf{u}} \\ p^\lambda \end{bmatrix} - \mu\|C-s\|_1\right) + \beta \alpha^2 \|p\|_2^2,
\end{align*}
where the term $\beta \alpha^2 \|p\|_2^2$ bounds the second derivative terms for some $\beta > 0$. Following a similar logic, we can obtain the bound in the other direction:
\begin{align*}
    &\phi(\mathbf{u}+\alpha p^{\mathbf{u}},\lambda+\alpha p^\lambda, s+\alpha p^s;\mu) - \phi(\mathbf{u},\lambda, s;\mu) \\
    &\geq \alpha\left(\nabla_{\mathbf{u},\lambda}\gamma\begin{bmatrix} p^{\mathbf{u}} \\ p^\lambda \end{bmatrix} - \mu\|C-s\|_1\right) - \beta \alpha^2 \|p\|_2^2.
\end{align*}
Dividing the inequality chain by $\alpha$ and taking the limit $\alpha\to 0$, we obtain the directional derivative:
\begin{align} \label{eq:merit_function_directional_derivative}
    D(\phi(\mathbf{u},\lambda,s;\mu), p^{\mathbf{u}}, p^\lambda) = \nabla_{\mathbf{u},\lambda}\gamma\begin{bmatrix} p^{\mathbf{u}} \\ p^\lambda \end{bmatrix} - \mu\|C-s\|_1
\end{align}
From \eqref{eq:merit_function_directional_derivative}, it should be clear that given $C-s \neq 0$, there exists a value for $\mu > 0$ such that the directional derivative is negative. As such, we propose the following expression to compute $\mu$, given some $\rho \in (0, 1)$:
\begin{align} \label{eq:merit_parameter}
    \mu \geq \left(\nabla_{\mathbf{u},\lambda}\gamma\begin{bmatrix} p^{\mathbf{u}} \\ p^\lambda \end{bmatrix}\right)/((1-\rho)\|C-s\|_1),
\end{align}
which results in $D(\phi(\mathbf{u},\lambda;\mu), p^{\mathbf{u}}, p^\lambda) \leq -\rho\mu\|C-s\|_1$.

In the case when $C-s = 0$, we unfortunately cannot guarantee that the directional derivative will always be negative. To see this, let us analyze the first term of \eqref{eq:merit_function_directional_derivative}:
\begin{align} \label{eq:stationarity_norm_directonal_derivative}
    \nabla_{\mathbf{u},\lambda}\gamma\begin{bmatrix} p^{\mathbf{u}} \\ p^\lambda \end{bmatrix} &= \nabla \mathcal{L}^\top \begin{bmatrix} L & G^\top \end{bmatrix} \begin{bmatrix} p^{\mathbf{u}} \\ p^\lambda \end{bmatrix} \nonumber \\
        &= \nabla \mathcal{L}^\top \begin{bmatrix} B + R & G^\top \end{bmatrix} \begin{bmatrix} p^{\mathbf{u}} \\ p^\lambda \end{bmatrix} \nonumber \\
        &= -\nabla \mathcal{L}^\top \nabla \mathcal{L} + \nabla \mathcal{L}^\top R p^\mathbf{u},
\end{align}
where $R$ denotes the residual matrix i.e. $R = L-B$ and we arrive at the third equality by plugging in for the stationarity condition from \eqref{eq:sqp_qp}. From \eqref{eq:stationarity_norm_directonal_derivative}, it should be immediately apparent that when the residual matrix is large and dominates the first term, it is possible for the directional derivative to be positive. This can be interpreted as a condition on how well the positive definite $B$ actually approximates the original stacked Hessian matrix $L$. For dynamic games where the agents have highly coupled and differing objectives, i.e. when $L$ is highly non-symmetric, it is reasonable to expect that the approximation would suffer and that we may not be able to achieve a decrease in the merit function. For this reason, we utilize a non-monotone strategy for the line search step, which will be discussed in the following.

\setlength{\textfloatsep}{0pt}
\begin{algorithm}[t!]
    \SetAlgoLined
	\KwIn{$\mathbf{u}$, $\lambda$, $s$, $p^\mathbf{u}$, $p^\lambda$, $p^s$, $P$, $\zeta$, $\tau$}
    $\phi \leftarrow \phi(\mathbf{u}, \lambda, s)$, $D\phi \leftarrow D(\phi(\mathbf{u},\lambda,s), p^{\mathbf{u}}, p^\lambda)$\;
    $\bar{\mathbf{u}}_0 \leftarrow \mathbf{u} + p^\mathbf{u}$, $\bar{\lambda}_0 \leftarrow \lambda + p^\lambda$, $\bar{s}_0 \leftarrow s + p^s$\;
    $\phi_0 \leftarrow \phi(\bar{\mathbf{u}}_0, \bar{\lambda}_0, \bar{s}_0)$\;
	\For{$i \in \{0, \dots, P-1\}$}{
        \If{$\phi_i \leq \phi + \zeta D\phi$}{
            \Return $\bar{\mathbf{u}}_i$, $\bar{\lambda}_i$\;
        }
	    $B_i, \ h_i, \ G_i, \ C_i \leftarrow$ \eqref{eq:sqp_approximation}\;
	    $p_i^\mathbf{u}, \ p_i^\lambda \leftarrow$ \eqref{eq:sqp_qp}, \eqref{eq:dual_step}\;
	    $\bar{\mathbf{u}}_{i+1} \leftarrow \bar{\mathbf{u}}_i + p_i^{\mathbf{u}}$, $\bar{\lambda}_{i+1} \leftarrow \bar{\lambda}_i + p_i^\lambda$\;
        $\bar{s}_{i} \leftarrow \min(0, C_i)$, $p_{i}^s \leftarrow C_i + G_i p_i^\mathbf{u} - \bar{s}_i$\;
        $\bar{s}_{i+1} \leftarrow \bar{s}_i + p_{i}^s$\;
        $\phi_{i+1} \leftarrow \phi(\bar{\mathbf{u}}_{i+1}, \bar{\lambda}_{i+1}, \bar{s}_{i+1})$\;
	}
    \If{$\phi_P \leq \phi + \zeta D\phi$}{
        \Return $\bar{\mathbf{u}}_P$, $\bar{\lambda}_P$\;
    }
    $B_P, \ h_P, \ G_P, \ C_P \leftarrow$ \eqref{eq:sqp_approximation}\;
    $p_P^\mathbf{u}, \ p_P^\lambda \leftarrow$ \eqref{eq:sqp_qp}, \eqref{eq:dual_step}\;
    $\bar{s}_{P} \leftarrow \min(0, C_P)$, $p_{P}^s \leftarrow C_P + G_P p_P^\mathbf{u} - \bar{s}_P$\;
    Find largest $\alpha \in (0, 1]$ such that $\phi(\bar{\mathbf{u}}_P + \alpha p_P^{\mathbf{u}}, \bar{\lambda}_P + \alpha p_P^{\lambda}, \bar{s}_P + \alpha p_P^{s}) \leq \phi_P + \zeta D\phi_P$\;
    $\bar{\mathbf{u}}_{P+1} \leftarrow \bar{\mathbf{u}}_P + \alpha p_P^{\mathbf{u}}$, $\bar{\lambda}_{P+1} \leftarrow \bar{\lambda}_P + \alpha p_P^{\lambda}$, $\bar{s}_{P+1} \leftarrow \bar{s}_P + \alpha p_P^{s}$\;
    $\phi_{P+1} \leftarrow \phi(\bar{\mathbf{u}}_{P+1}, \bar{\lambda}_{P+1}, \bar{s}_{P+1})$\;
    \eIf{$\phi_{P+1} \leq \phi + \zeta D\phi$}{
        \Return $\bar{\mathbf{u}}_{P+1}$, $\bar{\lambda}_{P+1}$\;
    }{
        Find largest $\alpha \in (0, 1]$ such that $\phi(\mathbf{u} + \alpha p^{\mathbf{u}}, \lambda + \alpha p^{\lambda}, s_P + \alpha p^{s}) \leq \phi + \zeta D\phi$\;
        \Return $\mathbf{u} + \alpha p^{\mathbf{u}}$, $\lambda + \alpha p^{\lambda}$
    }
	\caption{Watchdog Line Search}
	\label{alg:watchdog_line_search}
\end{algorithm}

\subsection{A Non-Monotone Line Search Strategy}

Line search methods are used in conjunction with merit functions to achieve a compromise between the goals of making rapid progress towards the optimal solution and keeping the iterates from diverging. This is done by finding the largest step size $\alpha \in (0, 1]$ such that the following standard decrease condition is satisfied \cite{boggs1995sequential}:
\begin{align} \label{eq:merit_decrease_condition}
    \phi(&\mathbf{u}_q+\alpha p_q^\mathbf{u}, \lambda_q+\alpha p_q^\lambda, s_q+\alpha p_q^s;\mu) \\
    &\leq \phi(\mathbf{u}_q, \lambda_q, s_q;\mu) + \zeta \alpha D(\phi(\mathbf{u}_q,\lambda_q,s_q;\mu), p_q^{\mathbf{u}}, p_q^\lambda), \nonumber
\end{align}
where $\zeta \in (0,0.5)$. However, since our merit function is not exact, the line search procedure can be susceptible to poor local minima which do not correspond to local GNE. We therefore include in our approach a non-monotone approach to line search called the \emph{watchdog} strategy \cite{conn2000trust}. Instead of insisting on a sufficient decrease in the merit function at every iteration, this approach allows for relaxed steps to be taken for a certain number of iterations, which can lead to increases in the merit function. The decrease requirement \eqref{eq:merit_decrease_condition} is then enforced after the prescribed number of relaxed iterations. The rationale behind this strategy is that we can use the relaxed steps as a way to escape regions where it is difficult to make progress w.r.t. the merit function. The algorithm is presented in Algorithm~\ref{alg:watchdog_line_search} where it is initialized with an iterate and step in the primal, dual, and slack variables. It additionally requires the parameters: $P \geq 1$ which is the maximum number of full steps to be taken, $\zeta$ for evaluating \eqref{eq:merit_decrease_condition}, and the backstepping coefficient $\tau$. In line 1, the merit function and its directional derivative are evaluated at the original iterate. In lines 2 and 3, we take the first full step and evaluate the merit function at the resulting iterate. In lines 4 to 17, the decrease condition w.r.t. to the original iterate is checked and the algorithm terminates if the condition is met. If not, a new step is computed at the current iterate and the full step is taken once again. After at most $P$ full steps, the decrease condition \eqref{eq:merit_decrease_condition} is enforced in lines 18 to 26, where line 21 corresponds to a backstepping line search procedure. In the case where the decrease condition is still not met, we return to perform a backstepping line search on the original iterate.

\section{The Dynamic Game SQP Algorithm}
\begin{algorithm}[t!]
    \SetAlgoLined
	\KwIn{$\mathbf{u}_0$, $\rho$, $\zeta$, $\tau$}
	$q \leftarrow 0$\;
	$\mathbf{u}_q \leftarrow \mathbf{u}_0$, $\lambda_q \leftarrow \max(0, -(G_0 G_0^\top)^{-1}G_0 h_0)$\;
	\While{not converged}{
	    $B_q, \ h_q, \ G_q, \ C_q \leftarrow$ \eqref{eq:sqp_approximation}\;
	    $p_q^\mathbf{u}, \ p_q^\lambda \leftarrow$ \eqref{eq:sqp_qp}, \eqref{eq:dual_step}\;
        $s_q \leftarrow \min(0, C_q)$, $p_{q}^s \leftarrow C_q + G_q p_q^\mathbf{u} - s_q$\;
	    \eIf{$C_q-s_q\neq 0$}{
            Compute $\mu$ from \eqref{eq:merit_parameter}\;
	    }{
	        $\mu \leftarrow 0$\;
	    }
	    $\mathbf{u}_{q+1}, \ \lambda_{q+1} \leftarrow$ watchdog line search\;
	    $q \leftarrow q + 1$
	}
	\Return $\mathbf{u}^\star \leftarrow \mathbf{u}_q$, $\lambda^\star \leftarrow \lambda_q$\;
	\caption{Dynamic Game SQP (DG-SQP)}
	\label{alg:dynamic_game_sqp}
\end{algorithm}

By combining the elements previously discussed, we arrive at the dynamic game SQP (DG-SQP) algorithm, which is presented in Algorithm~\ref{alg:dynamic_game_sqp}. The algorithm requires as input initial guesses of open-loop input sequences for each agent and the parameters $\rho$, $\zeta$, $\tau$. The parameter $\rho$ is used to compute the merit function parameter $\mu$ from \eqref{eq:merit_parameter}. $\zeta$ and $\tau$ are parameters for the watchdog linesearch procedure. Line 2 initializes the primal and dual iterates, where the dual variables are initialized as the least squares solution to \eqref{eq:stationarity}. Lines 3 to 14 perform the SQP iteration which has been described in Sections~\ref{sec:sqp_approach} and \ref{sec:practical_considerations}. An iterate is said to have converged to a local GNE if it satisfies the KKT conditions described in \eqref{eq:kkt} up to some user specified tolerance. Namely, for some given $\epsilon_1, \ \epsilon_2, \ \epsilon_3 > 0$, we require the conditions $\|\nabla\mathcal{L}(\mathbf{u}_q, \lambda_q)\|_\infty \leq \epsilon_1, \ \|C(\mathbf{u}_q)\|_\infty \leq \epsilon_2, \ |\lambda_q^\top C(\mathbf{u}_q)| \leq \epsilon_3$ be satisfied in order for the algorithm to terminate successfully. We additionally allow the algorithm to terminate with relative tolerance if the difference between successive primal and dual iterates are below a given threshold for a prescribed number of iterations. Note that when this occurs, the solution may not be a local GNE. The algorithm outputs the open-loop strategies for the $M$ agents and the corresponding Lagrange multipliers.

\section{Simulation Study}

In this section, we use simulation studies to demonstrate the performance of our DG-SQP algorithm and to compare our approach with the state-of-the-art GNE solver ALGAMES \cite{cleac2020algames}. We first use a ramp merge scenario to evoke a direct comparison with results reported in \cite{cleac2020algames}. We then compare the performance of the two approaches in a head-to-head car racing example to illustrate the advantages of DG-SQP. The DG-SQP algorithm was implemented in Python and can be found at \url{https://github.com/zhu-edward/DGSQP}. All results were obtained on a desktop with a 2.5 GHz 11th-Gen Intel Core i7 CPU.

\begin{figure}[t] 
    \vspace{-0.4cm}
    \centering
    \includegraphics[width=0.99\columnwidth]{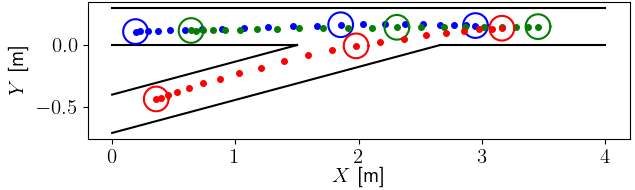}
    \vspace{-0.6cm}
    \caption{Example of a local GNE for three agents in the ramp merge scenario. The circles represent the collision avoidance constraints and the black lines indicate the road boundaries.}
    \label{fig:merge_sol}
\end{figure}

\subsection{Ramp Merge Scenario}
We first consider the ramp merge scenario which was investigated in \cite{cleac2020algames}. This scenario, illustrated in Fig.~\ref{fig:merge_sol}, involves three agents, where two are positioned on a straight road segment and the third on the ramp. The objective of the scenario is for the agent on the ramp to merge into traffic between the other two while avoiding collisions and respecting the road boundaries. We use a horizon length of $N = 20$. Like in \cite{cleac2020algames}, we conduct a Monte Carlo experiment of 1000 trials where the initial conditions of the three agents are randomly perturbed about some nominal values. As we attempted to recreate the experiment as closely as possible, we refer the reader to \cite{cleac2020algames} for additional details. Out of the 1000 trials, we observed that 99.3\% reached convergence with DG-SQP after an average of 12 iterations. This is comparable to the 99.5\% reported in \cite{cleac2020algames}, which required an average of 9 iterations of ALGAMES. Additionaly, for the successful trials, we observed an average solve time of 0.46s, and that 96\% converged to a solution within 0.8s. The results of this study show that our approach achieves similar performance in terms of solver success rate for the previously investigated ramp merge scenario. We note that while our execution times are certainly slower than the reported times for ALGAMES, we attribute this difference primarily to the fact that ALGAMES was implemented in Julia, which is a compiled language. The advantages of our approach will be clearly illustrated in the following car racing scenario.

\subsection{Head-to-head Car Racing Scenario}
We next demonstrate our DG-SQP algorithm in the context of head-to-head car racing, where our approach would be used by a vehicle to simultaneously obtain an open-loop control sequence and predictions of its opponent's behavior. We assume that no information about future plans is shared between agents and that they must avoid collisions with each other while also remaining within the boundaries of the track. In our examples, the agents are described by the kinematic bicycle model \cite{kong2015kinematic} with the following state and input vectors:
\begin{align*}
    x = [p_x, p_y, v_x, e_\psi, s, e_y]^{\top} \in \mathbb{R}^6, \ u = [a, \delta]^\top \in \mathbb{R}^2,
\end{align*}
where $p = (p_x, p_y)$ are the Cartesian position and $v_x$ is the longitudinal velocity of the vehicle's center of gravity (CoG). The remaining states are expressed in a Frenet reference frame \cite{micaelli1993trajectory} which is defined w.r.t. the centerline of the track. $s$ and $e_y$ denote the distance travelled along the track and lateral deviation from the centerline, of the CoG, respectively. $e_\psi$ is the deviation between the vehicle's heading and the tangent angle of the track centerline at $s$ (see \cite{rosolia2019learning} for detailed expressions). We obtain the discrete-time dynamics using Euler discretization with a time step of $T_s = 0.1$ s. The inputs to the vehicle are the longitudinal acceleration $a$ and front wheel steering angle $\delta$. By modeling the vehicle in this way, we can define the collision avoidance and track boundary constraints using the simple expressions
\begin{align*}
    (r^i+r^j)^2 - \|p^i-p^j\|_2^2 \leq 0, \ -W/2 \leq e_y^i \leq W/2,
\end{align*}
where $r^i$ and $r^j$ are the radii of the circular collision buffers for agents $i$ and $j$ respectively and $W$ is the width of the track. We subject the input magnitude and rate to identical box constraints for all agents.

\subsubsection{Comparison with ALGAMES}

\begin{figure}[t] 
    \centering
    \includegraphics[width=0.9\columnwidth]{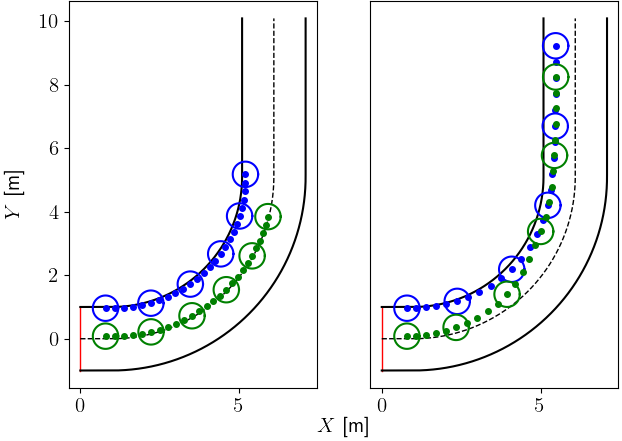}
    \vspace{-0.1cm}
    \caption{Example of the initial guess (left) and a corresponding local GNE (right) of horizon length $N=25$ for two agents on a curved track segment with a 90\degree \ turn. The circles represent the collision avoidance constraints.}
    \label{fig:init_v_sol}
    \vspace{-0.3cm}
\end{figure}
\begin{figure}[t] 
    \centering
    \includegraphics[width=0.99\columnwidth]{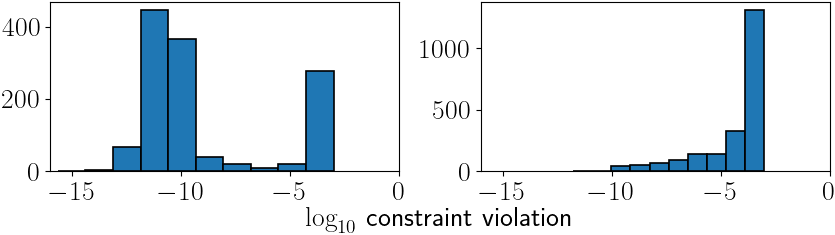}
    \vspace{-0.6cm}
    \caption{Largest non-zero constraint violation at solver convergence of DG-SQP (left) and ALGAMES (right).}
    \label{fig:converged_feasibility}
\end{figure}
\begin{figure}[t] 
    \centering
    \includegraphics[width=0.99\columnwidth]{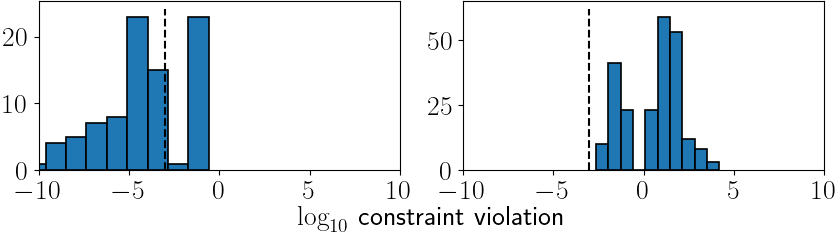}
    \vspace{-0.6cm}
    \caption{Largest non-zero constraint violation at solver failure of DG-SQP (left) and ALGAMES (right). The dashed line corresponds to $10^{-3}$.}
    \label{fig:failed_feasibility}
    \vspace{-0.4cm}
\end{figure}
\begin{figure}[t] 
    \centering
    \includegraphics[width=0.99\columnwidth]{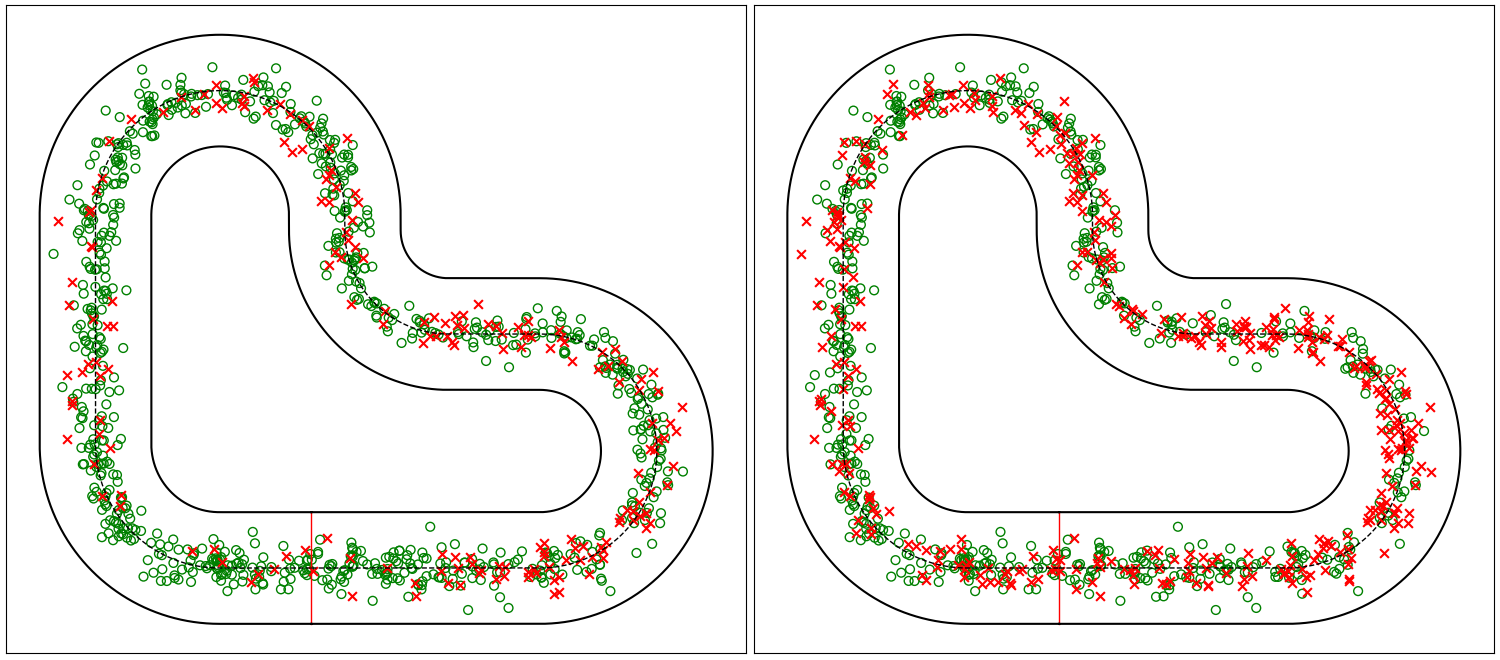}
    \vspace{-0.5cm}
    \caption{Monte Carlo results on a racetrack for DG-SQP (left) and ALGAMES (right). Each point corresponds to the average of the sampled initial position for the two agents. {\color{Green} $\circ$}, {\color{Red} $\times$} denote successful and failed trials respectively.}
    \label{fig:race}
\end{figure}
\setlength{\tabcolsep}{2pt}
\renewcommand{\arraystretch}{0.7}
\setlength{\heavyrulewidth}{1.5pt}
\begin{table*}[t]
\centering
\caption{Monte Carlo results for $M=2$ agents on a curved track with $\theta$ degree turns and game horizon length $N$. The first and second row for the \# Conv, \# Fail, \# Max Iters. Reached, and \# Solves statistics correspond to DG-SQP and ALGAMES respectively. The Time statistic is reported only for DG-SQP. For the Time statistic, the mean and (standard deviation) in seconds over all successful trials are reported.}
\vspace{-0.2cm}
\begin{tabular}{c|*{4}{c}|*{4}{c}|*{4}{c}} 
$\theta$ & \multicolumn{4}{c|}{45} & \multicolumn{4}{c|}{75} & \multicolumn{4}{c}{90} \\
$N$ & 10 & 15 & 20 & 25 & 10 & 15 & 20 & 25 & 10 & 15 & 20 & 25 \\
\midrule[1pt]
\multirow{2}{*}{\# Conv} & \textbf{200} & \textbf{199} & \textbf{199} & \textbf{185} & \textbf{200} & \textbf{199} & \textbf{195} & \textbf{169} & \textbf{200} & \textbf{197} & \textbf{193} & \textbf{172} \\
 & \textbf{200} & \textbf{199} & 189 & 172 & \textbf{200} & 192 & 180 & 136 & 198 & 193 & 167 & 142 \\
\midrule[0.5pt]
\multirow{2}{*}{\# Fail} & 0 & 0 & 0 & 0 & 0 & 0 & 0 & 0 & 0 & 0 & 0 & 0 \\
 & 0 & 0 & 9 & 16 & 0 & 6 & 16 & 39 & 1 & 4 & 24 & 43 \\
\midrule[0.5pt]
\# Max Iters. & 0 & 1 & 1 & 15 & 0 & 1 & 5 & 31 & 0 & 3 & 7 & 28 \\
Reached & 0 & 1 & 2 & 12 & 0 & 2 & 4 & 25 & 1 & 3 & 9 & 15 \\
\midrule[0.5pt]
\multirow{2}{*}{\# Solves} & 3 & 6 & 14 & 17 & 3 & 9 & 21 & 25 & 4 & 9 & 25 & 17 \\
 & 9 & 16 & 28 & 36 & 12 & 23 & 39 & 64 & 12 & 28 & 57 & 77 \\
\midrule[0.5pt]
Time [s] & 0.02(0.01) & 0.06(0.02) & 0.24(0.14) & 0.49(0.46) & 0.03(0.01) & 0.09(0.08) & 0.44(0.26) & 0.89(0.51) & 0.03(0.02) & 0.10(0.07) & 0.41(0.19) & 0.47(0.30)
\end{tabular}
\label{tab:mc_results}
\vspace{-0.6cm}
\end{table*}

We compare DG-SQP against a custom Python implementation of ALGAMES, which has also been made available in our source code. In this comparison study, we assume that all agents have identical dynamics and use the following cost functions:
\begin{subequations} \label{eq:cost_fns}
\begin{align}
    &l_k^i(x_k, u_k^i) = \frac{1}{2}{u_k^i}^\top R^i u_k^i + \frac{1}{2}{\Delta u_k^i}^\top R_d^i \Delta u_k^i \label{eq:comp_stage} \\
    &l_N^i(x_N) = -c_p s_N^i + c_c \sum_{j \neq i}\arctan(s_N^i-s_N^j), \label{eq:comp_term}
\end{align}
\end{subequations}
where $\Delta u_k^i = u_k^i-u_{k-1}^i$. The stage cost \eqref{eq:comp_stage} penalizes the input magnitude and rate with $R^i, R_d^i \succ 0$. The terminal cost \eqref{eq:comp_term} captures the competitive nature of racing with $c_p, c_c > 0$, where the first term encourages progress along the track and the second term is made small when agent $i$ is ahead of all other agents. Note that the dynamics considered in this example are made more complex by the addition of the Frenet frame states, which depend on the curvature of the track. 

We conduct the first comparison on a segment of curved track to examine the effect of track curvature and game horizon length on the performance of the GNE solvers. In particular, we perform a Monte Carlo analysis by randomly sampling feasible initial states near the start of the track segment and roll out the initial guess via a PID controller which maintains the car's speed and lateral deviation from the centerline. This initial guess is then used to initialize both algorithms. An example of this can be seen in the left plot of Fig.~\ref{fig:init_v_sol}. We do 200 trials each on tracks with 45, 75, and 90 degree turns for horizon lengths of 10, 15, 20, and 25. The results are presented in Tab.~\ref{tab:mc_results}. A trial is said to be successful if the iterates converge to a point where the KKT conditions \eqref{eq:kkt} are satisfied with tolerance $10^{-3}$. In Fig.~\ref{fig:converged_feasibility}, it can be seen that when successful, both algorithms return solutions which satisfy the constraints. 

Looking at the number of successful trials in the first row of Tab.~\ref{tab:mc_results}, it is clear that DG-SQP and ALGAMES perform similarly well for track segments with low curvature. These scenarios are similar to those investigated in \cite{cleac2020algames}, where dynamic games are played out on straight sections of road, and our results appear to further corroborate the high success rate of ALGAMES in these situations. However, as the curvature of the track or horizon length increases, we see that, while both solvers experience failures, DG-SQP outperforms ALGAMES in terms of success rate, showing a 21\% improvement in the case of $\theta=90\degree$ and $N=25$. In row four of Tab.~\ref{tab:mc_results} we report the average number of solves for the two approaches. For DG-SQP and ALGAMES, this corresponds to the number of QPs and linear systems that are solved, respectively. 
The average run time of successful trials of DG-SQP (in seconds) is shown in the fifth row. From this, we see that for horizons of moderate length, our Python implementation can achieve near real-time performance. We believe that an efficient implementation of DG-SQP in a compiled language would achieve a similar level of real-time performance as reported for ALGAMES in \cite{cleac2020algames}.

We turn next to the failure cases for the two solvers, where we have split them into the two categories \emph{\# Fail} in row two and \emph{\# Max Iters. Reached} in row three. \emph{\# Fail} counts the trials of ALGAMES, where the iterations diverged, which is defined as $\|\nabla\mathcal{L}(\mathbf{u}_q, \lambda_q)\|_\infty > 10^5$. For DG-SQP it additionally counts trials where the QP in \eqref{eq:sqp_qp} returned infeasible. \emph{\# Max Iters. Reached} counts the trials where a solver reached 50 iterations without converging. From these results, it is clear that the majority of the ALGAMES failure cases (68.1\%) were due to the solver diverging, whereas DG-SQP fails primarily due to reaching the maximum allowable iterations. We observe that in many of the DG-SQP failure cases, the iterates exhibit oscillatory behavior and fail the stationarity requirement of \eqref{eq:stationarity}. However, it is important to note that even in failure, most of the DG-SQP iterates remain feasible at termination (up to a tolerance of $10^{-3}$), whereas for ALGAMES, they do not. This is clearly seen in Fig.~\ref{fig:failed_feasibility}, which shows the distribution of constraint violation for the iterates returned by failed trials. We believe that this is due to the explicit linearized inequality constraints in DG-SQP. In ALGAMES, the iterates may not satisfy any form of the inequality constraint (linearized or otherwise), especially when the estimate of the Lagrange multipliers is poor. In practice, this is important as it is likely that in the case of failure, the outputs of DG-SQP, though not local GNE, can still be used as feasible solutions to the dynamic game.

In order to investigate the performance of the solvers in a more diverse set of racing conditions, we conduct a second comparison on the racetrack shown in Fig.~\ref{fig:race}. For this study, we fix the game horizon to $N=15$ and randomly sample the initial states of two agents such that an interaction is likely to occur over the horizon. Specifically, the agents start within 1.2 car lengths of each other and are traveling in the counter clockwise direction at velocities that exhibit at most a 25\% difference. Out of the 1000 trials, we observe that DG-SQP was successful in 814 trials whereas ALGAMES was successful in 618 trials. This is a 31.7\% improvement in success rate. Of the trials where the solvers failed, DG-SQP (ALGAMES) saw 13 (203) instances of divergence or QP failure and 185 (179) instances of reaching max iterations. Fig.~\ref{fig:race} illustrates the distribution of solver outcomes w.r.t. the averaged initial positions of the two agents. We note that failure cases for DG-SQP occur more frequently in sections of the racetrack with high curvature. This mirrors the trend observed in the previous study and highlights a limitation of our approach which will be investigated in future work.

\subsubsection{Effect of Merit Function and Non-Monotone Strategy}

\begin{figure}
    \centering
    \includegraphics[width=0.99\columnwidth]{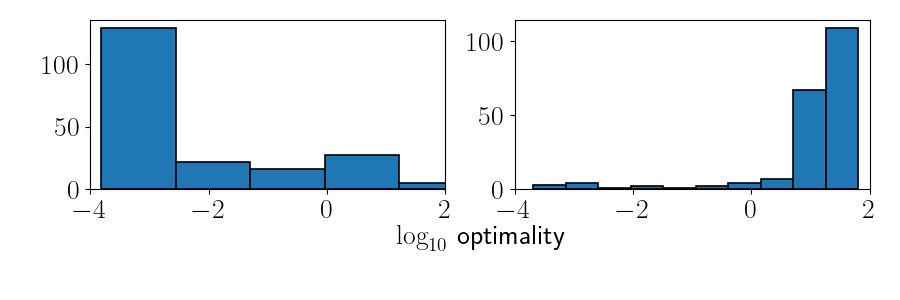}
    \vspace{-1.0cm}
    \caption{Stationarity, i.e. $\|\nabla\mathcal{L}(\mathbf{u}_q, \lambda_q)\|_\infty$, of iterate at termination of DG-SQP using the non-monotone line search with the novel merit function from Section~\ref{sec:practical_considerations} (left) and standard backstepping line search with the merit function from \cite{laine2021computation} (right).}
    \label{fig:ablation_results}
\end{figure}

We next investigate the effect of the merit function and non-monotone line search strategy proposed in Section~\ref{sec:practical_considerations}. We are especially interested in the case where the agent objectives are different as this induces additional asymmetry in the matrix $L_q$, which adversely affects the quality of our symmetric approximation $B_q$ in \eqref{eq:sqp_qp}. For this study, we consider the two agent scenario where the car in front (Agent 1) would like to block the one behind (Agent 2) and impede their progress. As such, we add the following term to the stage and terminal cost functions in \eqref{eq:cost_fns} for Agent 1:
\begin{align*}
    (1/2)c_b(e_{y,k}^1-e_{y,k}^2)^2,
\end{align*}
with $c_b > 0$ and keep the cost function for Agent 2 the same. This additional term encourages Agent 1 to match the lateral position of Agent 2 on the track. We compare the terminal value of $\|\nabla\mathcal{L}(\mathbf{u}_q, \lambda_q)\|_\infty$ for our approach with a variant of DG-SQP which uses a traditional backstepping line search and the merit function from \cite{laine2021computation}. The results are shown in Fig.~\ref{fig:ablation_results} for 200 trials on a track with $\theta=90\degree$ and $N=25$, where it is clear that our proposed modifications greatly improve the quality of the solution at termination of the DG-SQP algorithm. We observe a median stationarity of $9.369\times 10^{-4}$, compared to $19.76$ when none of the proposed modifications are used. 

\section{Future Work}

In this work, we have presented an SQP based approach to the solution of open-loop generalized Nash equilibria for dynamic games, which has provable local convergence and good practical performance stemming from novel modifications to the standard SQP algorithm from \cite{laine2021computation}. We have shown that, in the context of head-to-head car racing, our approach outperforms the state-of-the-art ALGAMES solver in terms of success rate. However, we observe that our approach can still suffer from convergence issues when dealing with problems with long horizons, highly nonlinear dynamics, or highly asymmetric cost. In future work, we hope to address these shortcomings by further investigating merit functions and line search strategies which can be tailored to dynamic games. We will also aim to implement the algorithm in a compiled language for real-time execution on a hardware platform.




\bibliographystyle{IEEEtran}
\bibliography{main.bib}

\end{document}